\documentclass{aims}
\usepackage{amsmath,amssymb,mathtools,latexsym,enumerate}
\usepackage{paralist}
\usepackage{graphics} 
\usepackage{epsfig} 
\usepackage{graphicx}
\usepackage[colorlinks=true]{hyperref}
\hypersetup{urlcolor=blue, citecolor=red}

\def\currentvolume{X}
 \def\currentissue{X}
  \def\currentyear{200X}
   \def\currentmonth{XX}
    \def\ppages{X--XX}
     \def\DOI{10.3934/xx.xx.xx.xx}

\newtheorem{theorem}{Theorem}[section]

\newtheorem{lemma}[theorem]{Lemma}

\theoremstyle{definition}
\newtheorem{definition}[theorem]{Definition}

\newcommand{\Mod}[1]{\ (\textup{mod}\ #1)}
\newcommand{\F}{\mathbb{F}}
\newcommand{\N}{\mathbb{N}}
\newcommand{\Z}{\mathbb{Z}}
\newcommand{\bv}{\mathbf{v}}
\newcommand{\bw}{\mathbf{w}}

\newcommand{\bu}{\mathbf{u}}
\newcommand{\bs}{\mathbf{s}}

\newcommand{\0}{\mathbf{0}}

\newcommand{\ba}{\mathbf{a}}
\newcommand{\bb}{\mathbf{b}}
\newcommand{\bc}{\mathbf{c}}
\newcommand{\bd}{\mathbf{d}}
\newcommand{\bx}{\mathbf{x}}
\newcommand{\by}{\mathbf{y}}

\newcommand{\etal}{\emph{et al. }}

\newcommand{\ie}{\emph{i.e.}}
\DeclarePairedDelimiterX{\inp}[2]{\langle}{\rangle}{#1, #2}
\DeclarePairedDelimiter\abs{\lvert}{\rvert}
\DeclareMathOperator{\lcm}{lcm}

\title[Skew Cyclic Codes Over $\F_4 R$] 
      {Skew Cyclic Codes Over $\F_4 R$}

\author[Benbelkacem, Ezerman, Abualrub, and Batoul]{}

\subjclass{Primary: 11T71, 11T06, 94B15}
 \keywords{to be added.}

 \email{benbelkacem213@gmail.com}
 \email{fredezerman@ntu.edu.sg}
 \email{abualrub@aus.edu}
 \email{a.batoul@hotmail.fr}

\thanks{$^*$ Corresponding author: N.~Benbelkacem}

\begin{document}
\maketitle

\centerline{\scshape Nasreddine Benbelkacem$^*$}
\medskip
{\footnotesize
 \centerline{Faculty of Mathematics, University of Science and Technology Houari Boumediene,}
\centerline{BP 32 El Alia, Bab Ezzouar, 16111 Algiers, Algeria}
}

\medskip

\centerline{\scshape Martianus Frederic Ezerman}
\medskip
{\footnotesize
\centerline{School of Physical and Mathematical Sciences, Nanyang Technological University,}
\centerline{21 Nanyang Link, Singapore 637371}
}

\medskip

\centerline{\scshape Taher Abualrub}
\medskip
{\footnotesize
	\centerline{Department of Mathematics and Statistics, College of Arts and Sciences,}
	\centerline{American University of Sharjah, P.O. Box 26666 Sharjah, United Arab Emirates}
}
\medskip

\centerline{\scshape Aicha Batoul}
\medskip
{\footnotesize
	\centerline{Faculty of Mathematics, University of Science and Technology Houari Boumediene,}
	\centerline{BP 32 El Alia, Bab Ezzouar, 16111 Algiers, Algeria}
}
\bigskip

 \centerline{(Communicated by the associate editor name)}

\begin{abstract}
This paper considers a new alphabet set, which is a ring that we call $\F_4R$, to construct linear error-control codes. Skew cyclic codes over the ring are then investigated in details. We define a nondegenerate inner product and provide a criteria to test for self-orthogonality. Results on the algebraic structures lead us to characterize $\F_4R$-skew cyclic codes. Interesting connections between the image of such codes under the Gray map to linear cyclic and skew-cyclic codes over $\F_4$ are shown. These allow us to learn about the relative dimension and distance profile of the resulting codes. Our setup provides a natural connection to DNA codes where additional biomolecular constraints must be incorporated into the design. We present a characterization of $R$-skew cyclic codes which are reversible complement.
\end{abstract}

\section{Introduction}\label{sec:intro}

The use of noncommutative rings to construct error control codes has recently been an active research area, initiated by the seminal works of Boucher \etal in~\cite{Boucher2007} and~\cite{Boucher2008} as well as that of Abualrub \etal in~\cite{Abualrub2010}. In the first two, Boucher and collaborators generalized the notion of cyclic codes by using generator polynomials in a noncommutative polynomial ring called the skew polynomial
ring. They supplied examples of skew cyclic codes with Hamming distances larger
than previously best-known linear codes of the same length and dimension.  In~\cite{Abualrub2010}, Abualrub \etal generalized the concept of skew cyclic codes to skew quasi-cyclic codes. They then constructed several new codes with Hamming distances exceeding the Hamming distances of the previously best-known linear codes with comparable parameters. 

Another emerging topic in the studies of error correcting codes is additive codes over mixed alphabets. Borges \etal introduced the class of $\Z_{2}\Z_{4}$-additive codes in~\cite{Borges2010}. This class generalizes binary and quaternary linear codes. A $\Z_{2}\Z_{4}$-additive code is defined to be a subgroup of $\Z_{2}^{r}\Z_{4}^{s}$. Abualrub \etal studied the structure of $\Z_{2}\Z_{4}$-additive cyclic codes in~\cite{Abualrub2014}. They determined the generator polynomials and the minimal generating sets for these codes.

In this paper, we merge the topic of skew cyclic codes with that of codes over mixed alphabets. In particular, we study the structure of linear skew cyclic codes over the ring $\F_{4}R$ where $\F_{4}$ is the finite field of four elements and $R=\{a + v b \mid a,b \in \F_{4}\}$ is the commutative ring with $16$ elements with $v^{2}=v$. Any codeword $\bc$ in a skew cyclic code $C$ over $\F_{4}R$ has the form $\bc=(a_{0},a_{1},...,a_{\alpha -1},b_{0},b_{1},...,b_{\beta -1})\in 
\F_{4}^{\alpha }R^{\beta }$. 

The followings are our contributions. 
\begin{enumerate}
	\item We show that the dual of a skew cyclic code over $\F_{4}R$ is also a skew cyclic code. In fact, skew cyclic codes over $\F_{4}R$ are left $R[X,\theta ]$-submodules of 
	$R_{\alpha ,\beta }=\F_{4}[X]/\langle X^{\alpha }-1\rangle \times R[X,\theta ]/\left\langle X^{\beta }-1\right\rangle$.
	 
	\item We determine their generator polynomials and establish interesting results that relate these codes to cyclic and quasi-cyclic codes over $\F_4 R$. First, we show that a skew cyclic code over $\F_{4}R$ is equivalent to an $\F_{4}R$-cyclic code if $\alpha$ and $\beta $ are both odd integers. Second, we establish that if $\alpha $ and $\beta $ are both even integers, then an $\F_{4}R$-skew cyclic code $C$ is equivalent to an $\F_{4}R$ quasi-cyclic code of index $2$. 
	
	\item Conditions for skew cyclic codes over $\F_{4}R$ to be self-orthogonal are studied.
	
	\item We use the Gray mapping to associate these codes to codes over $\F_{4}$ of length $\alpha +2 \beta$ and exhibit a nice relationship between these codes and their images over $\F_{4}$. The Gray image of any skew cyclic code over $\F_{4}R$ is the product of a cyclic code over $\F_{4}$ of length $\alpha$ and two skew cyclic codes, each of length $\beta$. We supply examples of skew cyclic codes over $\F_{4}R$ and their respective Gray images for different lengths. 
	
	\item Applications of these codes to DNA computing are included in our treatment.

\end{enumerate}

\section{Preliminaries}\label{sec:prelims}
	
Let $\F_{4}=\{0,1,w,w^{2}=w+1\}$ and $R=\{a+ vb \mid a,b$ in $\F_{4}\}$ be, 
respectively, the finite field with four elements and the commutative ring with $16$ elements where $v^{2}=v$. It is well-known that $R$ is a finite non-chain ring with two maximal ideals $\langle v\rangle$ and $\langle v+1\rangle$, making $R/\langle v\rangle $ and $R/\langle v+1\rangle $ isomorphic to $\F_{4}$. The Chinese Remainder Theorem then implies that $R=\langle v\rangle \times \langle v+1\rangle$. As was shown in~\cite{Bayram}, $R$ can be uniquely expressed as $\{a + vb = (b+a) v + a (v+1) \mid a,b\in \F_{4} \}$.
	
Let $A\oplus B=\{a+b\mid a\in A,b\in B\}$ and $A\otimes B=\{( a,b) \mid a\in A,b\in B\}$ as defined in~\cite{ZWS10}. An $\F_4$-linear code of length $n$ is a subspace of $\F_4^n$. A subset $C$ of $R^n$ is a \textit{linear code over $R$} if $C$ is an $R$-submodule. Given a linear
code $C$ over $R$, let 
\begin{align*}
C_{1} & \triangleq \{ \mathbf{x} + \mathbf{y} \in \F_{4}^{n} \mid (\mathbf{x}+ 
\mathbf{y}) v + \mathbf{x} (v+1) \in C \mbox{ for some } \mathbf{x}, 
\mathbf{y} \in \F_{4}^{n}\} \mbox{ and } \\
C_{2} & \triangleq \{\mathbf{x} \in \F_{4}^{n} \mid (\mathbf{x} + \mathbf{y}) v + 
\mathbf{x} (v+1) \in C \mbox{ for some } \mathbf{y} \in \F_{4}^{n}\}.
\end{align*}
One can quickly verify that $C_{1}$ and $C_{2}$ are linear codes over $\F_{4}$. In fact, any linear code $C$ over $R$ can be expressed as $C=v C_{1} \oplus (v+1) C_{2}$. Let $r = a + v b \in R$ and $\bc =(c_1,c_2,\ldots,c_n) \in C$, \emph{i.e.}, $c_j=a_j + v b_j$ with $a_j,b_j \in \F_4$ for $1 \leq j \leq n$. The $j$-th entry of $r \bc$ is 
\begin{align*}
(a + v b) (a_j + v b_j) &= ((a+b) v + a (v+1))(a_j + v b_j) \\
&= \underbrace{a a_j}_{x} + v (\underbrace{a b_j + b a_j + b b_j}_{y}) =
(x+y) v + x (v+1).
\end{align*}
Hence, $r \bc $ can be written in terms of $C_1$ and $C_2$ with 
\begin{equation*}
\mathbf{x} = a (a_1, a_2, \ldots, a_n) \mbox{ and } \mathbf{y} = (a+b)(b_1,b_2, \ldots, b_n) + b (a_1, a_2, \ldots, a_n).
\end{equation*}

\begin{definition}
Let an {\it automorphism} $\theta$ over $R$ be defined by
\begin{equation}\label{eq:theta}
\theta : R \mapsto R \mbox{ sending } a+vb \mapsto a^{2}+(v+1)~b^{2}.
\end{equation}
Restricted to $\F_4$, it interchanges $w$ and $w^2$ while keeping $\{0,1\}$ fixed. Note that our $\theta$ here is equal to the composition of automorphisms $\varphi \circ \theta_1$ in~\cite{Gursoy2014}. A subset $C$ of $R^{n}$ is said to be an {\it $R$-skew cyclic code} of length $n$ if two conditions are satisfied.
\begin{enumerate}
	\item $C$ is an $R$-submodule of $R^{n}$.
	\item If $\bc=\ (c_{0},c_{1},...,c_{n-1})\in C$ then the \textit{skew cyclic shift} of $\bc$ over $R$, denoted by $T_{\theta }(\bc) 
	\triangleq (\theta(c_{n-1}),\theta (c_{0}),...,\theta (c_{n-2}))$, must also be in $C$.
\end{enumerate}
\end{definition}

It is often convenient to associate a vector $\mathbf{a} = (a_0,a_1,\ldots,a_{n-1})$ with a polynomial $a(X):=a_0+a_1 X + \ldots + a_{n-1} X^{n-1}$ in an indeterminate $X$. This allows for constructions of codes using results from the algebra of polynomial rings.
	
The next two theorems can be inferred by a slight modification on the
corresponding theorems in~\cite{Gursoy2014} with $q$ restricted to $4$. The respective proof is therefore omitted for brevity.
	
\begin{theorem}(From \cite[Theorem 3]{Gursoy2014}) 
Let $C=v C_{1} \oplus (v+1) C_{2}$ be a linear code over $R$. Then $C$ is an $R$-skew cyclic code if and only if $C_{1}$ and $C_{2} $ are skew cyclic codes over $\F_{4}$.
\end{theorem}
	
\begin{theorem}\label{thm:gen}(From \cite[Theorem 5]{Gursoy2014}) 
Let $C=v C_{1} \oplus (v+1) C_{2}$ be a skew cyclic code of
length $n$ over $R$. Let $g_{1}(X)$ and $g_{2}(X)$ be the respective generator polynomials of $C_{1}$ and $C_{2}$ as $\F_4$-skew cyclic codes. Then $C=\langle v g_{1}(X) + 
(v+1) g_{2}(X) \rangle$.
\end{theorem}
	
For any element in $R$, we introduce a new ring homomorphism 
\begin{equation}\label{eq:ringhomo}
\eta : R \mapsto \F_{4} \mbox{ sending } a+vb \mbox{ to } a.
\end{equation}
Let $\F_{4}R := \{(a,b) \mid a \in \F_{4} \mbox{ and } b\in
R\}$. It is straightforward to verify that $\F_{4}R$ is an $R$-module under the multiplication 
\begin{equation}  \label{eq:ast}
d \ast (a,b)=(\eta ( d) a,db) \mbox{ with } d \in R \mbox{ and } (a,b)\in 
\F_{4}R.
\end{equation}
This extends naturally to $\F_{4}^{\alpha } R^{\beta}$. Let $\mathbf{x}=(a_{0},a_{1},...,a_{\alpha -1},b_{0},b_{1},...,b_{\beta-1})\in \F_{4}^{\alpha }R^{\beta}$, for $\alpha$ and $\beta \in \mathbb{N}$, and $d\in R$. Then 
\begin{equation}  \label{eq:scaprod}
d \ast \mathbf{x} = (\eta ( d) a_{0},\eta ( d) a_{1},...,\eta ( d) a_{\alpha-1},d b_{0},d b_{1},...,d b_{\beta -1}).
\end{equation}
	
\begin{definition}
A nonempty subset $C$ of $\F_{4}^{\alpha }R^{\beta }$ is called an 
{\it $\F_{4}R$-linear code} if it is an $R$-submodule of $\F_{4}^{\alpha }R^{\beta}$ with respect to the scalar multiplication in Equation (\ref{eq:scaprod}).
\end{definition}
A nondegenerate \textit{inner product} between $\mathbf{x}=(a_{0},a_{1},\ldots,a_{\alpha-1}, b_{0},b_{1},\ldots,b_{\beta-1})$ and 
$\mathbf{y}=(d_{0},d_{1},\ldots,d_{\alpha-1}, e_{0},e_{1},...,e_{\beta -1})$
is given by 
\begin{equation}  \label{eq:inprod}
\langle \mathbf{x}, \mathbf{y}\rangle =v\sum_{i=0}^{\alpha-1}a_{i}d_{i}+\sum_{j=0}^{\beta -1}b_{j}e_{j}\in R.
\end{equation}
The dual code of an $\F_4 R$-linear code $C$, denoted by $C^{\perp }$, is also $\F_{4}R$-linear and is defined in the usual way as 
\begin{equation*}
C^{\perp }:=\{\mathbf{y} \in \F_{4}^{\alpha }R^{\beta } \mid \langle 
\mathbf{x}, \mathbf{y} \rangle = 0 \mbox{ for all } \mathbf{x} \in C\}.
\end{equation*}
	
Let 
\begin{align*}
a(X) & =a_{0} + a_{1} X + \ldots + a_{\alpha -1} X^{\alpha -1} \in \F_{4}[X]/\langle X^{\alpha }-1\rangle \mbox{ and } \\
b(X) & = b_{0} + b_{1} X + \ldots + b_{\beta-1} X^{\beta -1} \in
R[X,\theta]/\langle X^{\beta }-1\rangle.
\end{align*}
Then any codeword $\bc=(a_{0},a_{1},...,a_{\alpha-1}, b_{0},b_{1},...,b_{\beta -1})\in \F_{4}^{\alpha }R^{\beta }$ can be
identified with a module element consisting of two polynomials such that 
\begin{equation}  \label{eq:c}
c(X)=(a(X),b(X)).
\end{equation}
This identification gives a one-to-one correspondence between $\F_{4}^{\alpha }R^{\beta }$ and 
\begin{equation}  \label{eq:corresp}
R_{\alpha ,\beta } \triangleq \F_{4}[X]/\langle X^{\alpha }-1\rangle
\times R[X,\theta ]/\langle X^{\beta }-1\rangle.
\end{equation}
Let $r(X)=r_{0} + r_{1} X+ \ldots + r_{t} X^{t} \in R[X,\theta]$ and $(a(X),b(X)) \in R_{\alpha ,\beta}$. Their product is 
\begin{equation}  \label{eq:mult}
r(X) \ast (a(X),b(X))=(\eta (r(X)) a(X) , r(X)b(X))
\end{equation}
where $\eta(r(X))= \eta ( r_{0} ) + \eta ( r_{1} ) X + \ldots + \eta (r_{t})
X^{t} \in \F_4[X]$. Here, $\eta(r(X))a(X)$ is the usual polynomial
multiplication in $\F_{4}[X]/\langle X^{\alpha }-1\rangle$ while $r(X) b(X)$ is the polynomial multiplication in $R[X,\theta ]/\langle X^{\beta }-1\rangle$ where $X (a + v b) = \left( a^2 + (v+1) b^2 \right) X$.
	
\begin{theorem}
$R_{\alpha,\beta}$ is a left $R[X,\theta ]$-module with respect to $\ast$ in Equation~(\ref{eq:mult}).
\end{theorem}
	
\begin{proof}
Verifying that the required properties are satisfied over $\F_{4}[X]/\langle X^{\alpha }-1\rangle$ is easy since we do not have to deal with skewness. Verifying over $R[X,\theta ]/\langle X^{\beta }-1\rangle$ is routine, albeit tedious. It suffices to use the facts that $\theta$ is a homomorphism with $\theta^{-1}=\theta$.
\end{proof}

\section{Generator Polynomials of $\F_{4} R$-Skew Cyclic Codes}
	
This section begins with a formal definition of an $\F_4R$-skew cyclic code and proposes a method to determine the generator polynomial of any $\F_{4}R$-skew cyclic code $C$ in $R_{\alpha,\beta}$. We say that two codes are \emph{equivalent} if one can be obtained from the other by some composition of a permutation of the first $\alpha$ positions, a permutation of the last $\beta$ positions, and multiplication of the symbols appearing in a chosen position by a nonzero scalar.

\begin{definition}
An $\F_4 R$-linear code $C$ of length $n=\alpha+\beta$ is said to be {\it $\F_{4}R$-skew cyclic} if, for any codeword $\bc=(a_{0},a_{1},\ldots,a_{\alpha -1},b_{0},b_{1}, \ldots, b_{\beta -1})\in C$, its skew cyclic shift 
$T_{\theta }(\bc) \triangleq (a_{\alpha -1},a_{0},\ldots,a_{\alpha -2},\theta
(b_{\beta -1}),\theta (b_{0}),\ldots,\theta (b_{\beta -2}))$ is also in $C$.
\end{definition}
	
\begin{theorem}
Let $C$ be an $\F_{4}R$-skew cyclic code of length $n=\alpha +\beta $
such that $\beta$ is an even integer. Then $C^{\perp }$ is also an $\F_{4}R$-skew cyclic code of the same length.
\end{theorem}
	
\begin{proof}
It suffices to show that, for any  
$\bx=(a_{0},a_{1},\ldots,a_{\alpha -1},b_{0},b_{1},\ldots,b_{\beta -1})\in C^{\perp }$, we have $T_{\theta}(\bx)\in C^{\perp }$. Let   
$\by=(d_{0},d_{1},\ldots,d_{\alpha -1},e_{0},e_{1},\ldots,e_{\beta -1})$ be any codeword in $C$. Then  
\begin{multline*}
\inp {T_{\theta} (\bx)}{\by} = \\
\langle ( a_{\alpha -1},a_{0},\ldots,a_{\alpha -2},\theta (b_{\beta-1}),\theta (b_{0}),\ldots,\theta (b_{\beta -2})),(d_{0},d_{1},\ldots,d_{\alpha -1},e_{0},e_{1}, \ldots , e_{\beta -1})\rangle \\
 =v(a_{\alpha -1}d_{0}+a_{0}d_{1}+\ldots+a_{\alpha -2}d_{\alpha -1})+(\theta(b_{\beta-1})e_{0}+\theta (b_{0})e_{1}+\ldots+
\theta (b_{\beta -2})e_{\beta-1}).
\end{multline*}
Hence, one only needs to show that 
\[
0 = a_{\alpha -1}d_{0}+a_{0}d_{1}+\ldots+a_{\alpha-2}d_{\alpha -1} \mbox{ and } 
0 = \theta (b_{\beta -1})e_{0}+\theta
(b_{0})e_{1}+\ldots+\theta (b_{\beta -2})e_{\beta -1}.
\]
Now, let $\gamma :=\lcm(\alpha ,\beta )$. Then $\gamma$ is an even integer since $\beta$ is an even integer. Since $C$ is $\F_{4}R$-skew cyclic, for any $\by \in C$ we have $T_{\theta }^{\gamma}(\by) = \by$ and $T_{\theta }^{\gamma -1}(\by) \in C$. Hence, $\inp{\bx}{T_{\theta }^{\gamma -1} (\by)} =0$. 
Since $T_{\theta }^{\gamma -1} (\by) = ( d_{1},\ldots,d_{\alpha -1},d_{0},\theta (e_{1}),\ldots,\theta(e_{\beta -1}),\theta (e_{0}))$, we then obtain 
\[
v \sum_{j=0}^{\alpha-1} a_j \, 
d_{(j+1) \Mod{\alpha}} + \sum_{j=0}^{\beta-1} b_j \, 
\theta \left(e_{(j+1) \Mod{\beta}}\right)=0.
\] 
This implies
\begin{align*}
0 & = a_{\alpha -1}d_{0}+a_{0}d_{1}+a_{1}d_{2}+\ldots+a_{\alpha-2}d_{\alpha -1} \mbox{ and }\\ 
0 & = b_{\beta -1}\theta (e_{0})+b_{0}\theta (e_{1})+b_{1}\theta(e_{2})+\ldots+b_{\beta -2}\theta (e_{\beta -1}).
\end{align*}
Applying $\theta$ to both sides of the last equation yields
\[
\theta(0)=\theta (b_{\beta -1})e_{0}+\theta (b_{0})e_{1}+\theta(b_{1})e_{2}+\ldots+\theta (b_{\beta -2})e_{\beta -1}=0,
\]
completing the proof. 
\end{proof}
	
\begin{theorem}
A code $C$ is $\F_{4}R$-skew cyclic if and only if $C$ is a left 
$R[X,\theta]$-submodule of $R_{\alpha ,\beta }$ under the multiplication $\ast$.
\end{theorem}
	
\begin{proof}
Let $c(X)=(a(X),b(X))$ be any codeword of an $\F_{4}R$-skew cyclic code $C$. Hence, $(a_{0},a_{1},\ldots,a_{\alpha-1},b_{0},b_{1},\ldots,b_{\beta-1})$ and all of it's $T_{\theta}$-skew cyclic shifts are in $C$. We associate, for each $j \in \N$, the polynomial
\begin{multline*}
X^{j}\ast c(X) =(a_{\alpha-j} + a_{\alpha-j+1} X + \ldots +
a_{\alpha-j-1} X^{\alpha -1}, \\
\theta^{j} (b_{\beta-j}) + \theta^{j} (b_{\beta-j + 1}) X + \ldots + \theta^{j} (b_{\beta-j-1}) X^{\beta-1})
\end{multline*}
with the vector 
\[
(a_{\alpha-j},a_{\alpha-j+1},\ldots,a_{\alpha-j-1}, 
\theta^{j}(b_{\beta-j}),\theta^{j}(b_{\beta-j+1}),\ldots,\theta^{j}(b_{\beta-j-1})).
\]
The indices of the first block (of length $\alpha$) are taken modulo $\alpha $ and those of the second block (of length $\beta$) are taken modulo $\beta$. By the $\F_4 R$-linearity of $C$, we have $r(X) \ast c(X) \in C$ for any $r(X) \in R[X,\theta]$. Thus, $C$ is a left $R[X,\theta ]$-submodule of $R_{\alpha,\beta}$.
		
Conversely, let $C$ be a left $R[X,\theta]$-submodule of the left $R[X,\theta]$-module $R_{\alpha,\beta }$. Then, for any $c(X) \in C$, we have $X^j \ast c(X) \in C$ for any $j \in \N$. Thus, $C$ is indeed an $\F_{4}R$-skew cyclic code. 
\end{proof}
	
Let $C$ be an $\F_{4}R$-skew cyclic code. Let $c(X)=(a(X),b(X))$ be an element in $C$. Let $\ell(X)$ be an element in $\F_{4}[X]/ \langle X^{\alpha}-1\rangle$. We use $\mathbf{0}$ to denote either the zero vector $(0,0,\ldots,0)$ or the zero polynomial. Let 
\begin{align*}
I & \triangleq \{b(X)\in R[X,\theta] /\langle X^{\beta}-1\rangle \mid
(\ell(X),b(X))\in C \} \mbox{ and } \\
J & \triangleq \{a(X)\in \F_{4}[X]/\langle X^{\alpha}-1\rangle \mid
(a(X),\mathbf{0}) \in C\}.
\end{align*}
The next results establish useful properties of the sets $I$ and $J$
	
\begin{lemma}\label{lem:J}
$J$ is an ideal in $\F_{4}[X]/\left\langle X^{\alpha }-1 \right \rangle$ generated by a left divisor of $X^{\alpha}-1$.
\end{lemma}
	
\begin{proof}
Let $a_{1}(X)$ and $a_{2}(X)$ be elements of $J$. By definition, $(a_{1}(X),\0)$ and $(a_{2}(X),\0)$ are in $C$. Hence, $(a_{1}(X),\0) + (a_{2}(X) ,\0) = (a_{1}(X) + a_{2}(X),\0) \in C$, making $a_{1}(X) +a_{2}(X) \in J$. Let $s(X) \in \F_{4}[X]/\left\langle X^{\alpha }-1 \right\rangle$ and $a(X) \in J$. Then $(a(X),\0)$ is in $C$. Because $C$ is a left $R[X,\theta]$-module, we have 
\[
s(X) \ast (a(X),\0)=(s(X) a(X),\0) \in C \implies s(X) a(X) \mbox{ modulo } (X^{\alpha}-1) \in J.
\]
Thus, $J$ is an ideal in $\F_{4}[X]/\left\langle X^{\alpha}-1 \right\rangle$ generated by a left divisor $f(X)$ of $X^{\alpha}-1$.
\end{proof}
	
\begin{lemma}\label{lem:I}
$I$ is a principally generated left $R[X,\theta]$-submodule of $R[X,\theta]/\langle X^{\beta}-1 \rangle$.
\end{lemma}
	
\begin{proof}
Let $b_{1}(X)$ and $b_{2}(X)$ be elements in $I$. Then there exist polynomials  $\ell_1(X)$ and $\ell_2(X)$ in $\F_4[X]/\langle X^{\alpha}-1 \rangle$ such that  $(\ell_{1}(X),b_{1}(X)),(\ell_{2}(X),b_{2}(X)) \in C$. Hence, 
\[(\ell_{1}(X),b_{1}(X)) +(\ell_{2}(X),b_{2}(X))=
(\ell_{1}(X) +\ell_{2}(X) ,b_{1}(X) +b_{2}(X)) \in C, 
\]
implying $b_{1}(X) +b_{2}(X) \in I$. Let $r(X)\in R[X,\theta]/\langle X^{\beta }-1 \rangle$ and $(\ell(X), b(X))\in C$. Since $C$ is a left $R[X,\theta]$-submodule of $R_{\alpha, \, \beta}$, we have  
\[
r(X) \ast (\ell(X),b(X)) = (\eta(r(X)) \ell(X) \mbox{ modulo} (X^{\alpha}-1), r(X) b(X) \mbox{ modulo } (X^{\beta}-1))
\]
in $C$, making 
$r(X) b(X) \mbox{ modulo } (X^{\beta}-1) \in I$. Thus, $I$ is a left submodule in 
$R[X,\theta]/ \langle X^{\beta }-1 \rangle$ and, by Theorem~\ref{thm:gen}, 
$I=\langle g(X) \rangle$ where 
\begin{equation}\label{eq:g}
g(X) \triangleq v g_{1} (X) + (v+1) g_{2}(X).
\end{equation}
\end{proof}
	
The following result classifies all $\F_4 R$-skew cyclic codes.
\begin{theorem}
Let $g(X)$ be as defined in Equation (\ref{eq:g}). Let $C$ be an $\F_{4}R$-skew cyclic code. Then $C$ is generated as a left submodule of $R_{\alpha,\beta}$ by $(f(X),\0)$ and $(\ell(X),g(X))$ where $\ell(X)$ is an element in $\F_{4}[X]/\langle X^{\alpha}-1 \rangle$ and $f(X)$ is a left divisor of $X^{\alpha}-1$.
\end{theorem}
	
\begin{proof}
Let $\bc=(\bc_{1},\bc_{2}) \in C$ with $\bc_1 \in \F_4^{\alpha}$ and $\bc_2 \in R^{\beta}$. Then $c_{2}(X) \in I$ and we write $c_{2}(X) = q(X) \, g(X)$ for some $q(X) \in R[X,\theta]/\langle X^{\beta }-1\rangle$. There exist $\ell(X) \in \F_4[X]/\langle X^{\alpha }-1\rangle$ such that  $(\ell(X),g(X)) \in C$ since $g(X) \in I$. We have 
\begin{align*}
\bc=(\bc_{1},\bc_{2}) &= (c_1(X),\0) + (\0, q(X) g(X)) \\
&= (c_1(X),\0) + (\eta(q(X)) \ell(X), q(X) g(X)) + (\eta(q(X)) \ell(X), \0)\\
&= (c_1(X),\0) + q(X) \ast \left((\ell(X),g(X)) + (\ell(X), \0)\right). 
\end{align*}
Hence, $(\eta(q(X)) \ell(X) + c_1(X), \0) \in C$, making $\eta(q(X)) \ell(X) + c_1(X) \in J$. By Lemma~\ref{lem:J}, there exists $p(X) \in J$ satisfying $\eta(q(X)) \ell(X) + c_1(X) = p(X) f(X)$. Thus, $c(X) = q(X) \ast (\ell(X), g(X)) + (p(X) f(X),\0)$.
\end{proof}
	
\begin{lemma}
Let $C$ be an $\F_{4}R$-skew cyclic code. Then, 
without loss of generality, we can assume $\deg (\ell(X)) < \deg (f(X))$.
\end{lemma}
	
\begin{proof}
Suppose that $\deg(\ell(X)) - \deg(f(X)) = k \geq 0$. Consider the code $D$ generated by the set
\[
\{(f(X),\0),(\ell(X),g(X)) + s X^{k} \ast (f(X),\0)\} =
\{(f(X),\0),(\ell_{1}(X),g(X))\}
\]
where $\ell_{1}(X) = \ell(X)+s X^{k} f(X)$ for some $s \in \F_4$. Hence, $D \subseteq C$. On the other hand, 
\[
(\ell(X),g(X)) =(\ell(X),g(X)) +s X^{k} \ast (f(X),\0) - s X^{k} \ast (f(X),\0).
\]
Hence, $C\subseteq D$, making $C=D$. Notice here that $\deg(\ell_{1}(X)) < \deg(\ell(X))$. We repeat the same process on $\ell_{1}(X)$ until we obtain $\deg(\ell(X)) < \deg(f(X))$. 
\end{proof}
	
\begin{theorem}
An $\F_{4}R$-skew cyclic code is equivalent to an 
$\F_{4}R$-cyclic code if both $\alpha$ and $\beta$ are odd integers.
\end{theorem}
	
\begin{proof}
Let $C$ be an $\F_{4}R$-skew cyclic code and 
$\gamma :=\lcm(\alpha,\beta)$. Then  $\gcd(\gamma,2)=1$ since $\gamma$ is odd. Then there exist integers $k$ and $j$ such that $\gamma k + 2 j=1$ and, hence, $2 j = 1 - \gamma k = 1 + \gamma t$ for some $t > 0$ where $t \equiv -k \pmod {\gamma}$. As in Equation (\ref{eq:c}), let 
$c(X)=(a(X), b(X)) \in C$. Then 
\begin{align*}
X^{2 j} \ast c(X) &= X^{2 j} \ast 
\left(\sum_{i=0}^{\alpha -1} a_{i} X^i  ,
\sum_{i=0}^{\beta -1}b_{i}X^i \right) =
\left(\sum_{i=0}^{\alpha-1} a_{i} X^{i+2 j} , 
\sum_{i=0}^{\beta-1} \theta^{2 j} (b_{i}) X^{i+2 j} \right)\\ 
&= \left(\sum_{i=0}^{\alpha-1} a_{i} X^{i+1+ \gamma t}, 
\sum_{i=0}^{\beta-1} \theta^{2 j} (b_{i}) X^{i+1+ \gamma t}\right)\\
& = \left(\sum\limits_{i=0}^{\alpha-2} a_{i} X^{i+1+\gamma t} + 
a_{\alpha-1} X^{\alpha +\gamma t} , 
\sum_{i=0}^{\beta-2} b_{i} X^{i+1+\gamma t}
+ a_{\beta-1} X^{\beta +\gamma t} \right) \\
&= \left(\sum_{i=0}^{\alpha-2} a_{i} X^{i+1} + a_{\alpha-1} ,
\sum_{i=0}^{\beta-2} b_{i} X^{i+1} + b_{\beta-1} \right).
\end{align*}
The second to the last equation is due to $\theta^2(r)=r$ for all $r \in R$ while the last equation follows because $X^{\alpha} =X^{\beta}=X^{\gamma}=1$.
\end{proof}
	
\begin{theorem}
An $\F_{4}R$-skew cyclic code is equivalent to an 
$\F_{4}R$-quasi-cyclic code of index $2$ if both $\alpha$ and $\beta$ are even integers.
\end{theorem}
	
\begin{proof}
Let $C$ be an $\F_{4}R$-skew cyclic code, $\alpha=2N$, and $\beta=2M$ for some $N,M \in \N$. Then $\gamma =\lcm(\alpha,\beta)$ is an even integer with $\gcd(\gamma,2)=2$. 
For any 
\[
\bc=(a_{0,0}, a_{0,1}, \ldots, a_{N-1,0}, a_{N-1,1}~,~b_{0,0}, b_{0,1}, \ldots, b_{M-1,0}, b_{M-1,1}) \in C
\]
there exist integers $k \geq 0$ and $j$ such that $2 j = 2 + k \gamma$. Consider 
\begin{align*}
&T_{\theta^{2+ k \gamma}} \left( a_{0,0},a_{0,1},\ldots,a_{N-1,0},a_{N-1,1} , 
b_{0,0},b_{0,1},\ldots,b_{M-1,0},b_{M-1,1}\right) =  \\
&T_{\theta^{k \gamma}}\left(a_{N-1,0},a_{N-1,1},\ldots,a_{N-2,0},a_{N-2,1} ,
b_{M-1,0},b_{M-1,1},\ldots,b_{M-2,0},b_{M-2,1}\right) = \\
& \left( a_{N-1,0},a_{N-1,1},\ldots,a_{N-2,0},a_{N-2,1},
b_{M-1,0},b_{M-1,1},\ldots,b_{M-2,0},b_{M-2,1}\right) \in C 
\end{align*}
since $T_{\theta^{k \gamma}}(\bc) = \bc$ for any $\bc \in 
\F_{4}^{\alpha }R^{\beta}$. 
Thus, $C$ is equivalent to an $\F_{4}R$-quasi cyclic code of length $n=\alpha +\beta$ and index $2$.
\end{proof}

\section{The Gray Mapping}

The classical Gray mapping $\phi^{\ast}: R \mapsto \F_4^2$ is defined by $\phi^{\ast}(a+vb)=(a+b,a)$ for any $a + v b \in R$. The \textit{Lee weight} of any element in $R$ is the Hamming weight of its image under $\phi^{\ast}$. This map extends naturally to vectors in $R^n$. For any $\mathbf{x}=(x_{0},x_{1},\ldots,x_{\alpha-1}) \in \F_{4}^{\alpha }$ and 
$\mathbf{y}=(y_{0},y_{1},\ldots,y_{\beta-1}) \in R^{\beta}$, the \textit{Gray map} over $\F_4 R$ is defined by 
\begin{equation*}
\phi : \F_{4}^{\alpha }R^{\beta } \mapsto \F_{4}^{\alpha+ 2
\beta} \mbox{ sends } (\mathbf{x},\mathbf{y}) \mbox{ to } (\mathbf{x},\phi^{\ast}(\mathbf{y})).
\end{equation*}
The map $\phi $ is an isometry which transforms the Lee distance in $\F_{4}^{\alpha }R^{\beta}$ to the Hamming distance in $\F_{4}^{\alpha + 2 \beta}$. For any $\F_{4}R$-linear code 
$\mathcal{C}$, the code $\phi(\mathcal{C})$ is $\F_4$-linear. Furthermore, we have 
\begin{equation} \label{eq:weight}
wt(\mathbf{x},\mathbf{y})=wt_{H}(\mathbf{x})+wt_{L}(\mathbf{y})
\end{equation}
where $wt_{H}(\mathbf{x})$ is the Hamming weight of $\mathbf{x}$ and 
$wt_{L}(\mathbf{y})$ is the Lee weight of $\mathbf{y}$.
	
\begin{theorem}
Let $C$ be a self-orthogonal $\F_{4}R$-linear code under the inner product defined in Equation~(\ref{eq:inprod}). Then $\phi(C)$ is a Euclidean self-orthogonal code over $\F_{4}$.
\end{theorem}
	
\begin{proof}
It suffices to show that the Gray images of codewords are Euclidean orthogonal
whenever the codewords are orthogonal. Let $\mathcal{C}$ be a self-orthogonal $\F_{4}R$-linear code of length $\alpha+\beta$. Let $\bv=(\ba,\bb + v \, \bc),  \bw=(\bd,\bu + v \, \bs) \in \F_{4}^{\alpha}\times R^{\beta}$ be codewords in $C$ with $\ba, \bd \in \F_{4}^{\alpha}$ and 
$\bb, \bc, \bu, \bs \in 
\F_{4}^{\beta}$. Then, by Equation (\ref{eq:inprod}),
\[
\left\langle \bv, \bw \right\rangle  
= v (\ba \cdot \bd) + \bb \cdot \bu + v (\bb \cdot \bs + \bc \cdot \bu + \bc \cdot \bs)= 0 + v 0 \in R.
\]
Hence, $\bb \cdot \bu =0$ and $\ba \cdot \bd + \bb \cdot \bs + 
\bc \cdot \bu + \bc \cdot \bs =0$ .
Since $\phi(\bv) =(\ba, \bb + \bc ,\bb)$ and 
$\phi(\bw) =(\bd, \bu + \bs , \bu)$, one gets   
\[
\phi(\bv) \cdot \phi(\bw) =
\ba \cdot \bd + \bb \cdot \bu + \bb \cdot \bs + 
\bc \cdot \bu + \bc \cdot \bs + \bb \cdot \bu =0. 
\]
Therefore, the code $\phi(C)$ is Euclidean self-orthogonal.
\end{proof}

\begin{theorem}\label{thm:product}
Let $C$ be an $\F_4 R$-skew cyclic code of length $n=\alpha +2 \beta$. Then, 
$\phi(C)=C_{0} \otimes C_{1}\otimes C_{2},$ where $C_{0}$ is a cyclic code of 
length $\alpha $ in $\F_4[X] /\langle X^{\alpha }-1\rangle $ and 
both $C_{1}$ and $C_{2}$ are skew cyclic codes of length $\beta$ in 
$R[X]/\left\langle X^{\beta}-1\right\rangle$. Moreover, 
$\displaystyle{\abs{\phi(C)} = \prod_{i=0}^{2} \abs{C_{i}}}$.
\end{theorem}

\begin{proof}
From $\{\bx=(a_{0},a_{1},\ldots,a_{\alpha-1},b_{0} + v c_{0}, b_{1} + v c_{1},\ldots, 
b_{\beta-1} + v c_{\beta-1}) : \bx \in C \}$, we construct the codes
\begin{multline*}
C_{0} := \{(a_{0},a_{1},\ldots,a_{\alpha-1}) \} , 
C_{1} := \{(b_{0} + c_{0}, b_{1} + c_{1},\ldots, b_{\beta-1} + c_{\beta -1}) \}, \\
\mbox{and } C_{2} := \{(b_{0}, b_{1},\ldots, b_{\beta -1})\}.
\end{multline*}
A codeword $\bu:=(a_{0},a_{1},\ldots,a_{\alpha -1}) \in C_{0}$ corresponds to a codeword 
\[
\bx=(a_{0},a_{1},\ldots,a_{\alpha-1},b_{0} + v c_{0}, b_{1} + v c_{1},\ldots, b_{\beta-1} + v c_{\beta-1}) \in C.
\]
Since $C$ is an $\F_4 R$-skew cyclic code, we know that $T_{\theta} (\bx)$ is given by
\[
(a_{\alpha-1}, a_{0}, a_{1}, \ldots,a_{\alpha-2}, \theta(b_{\beta-1} + v c_{\beta-1}), 
\theta(b_{0} + v c_{0}),\ldots,\theta(b_{\beta-2} + v c_{\beta-2})) \in C.
\]
Hence, $\left( a_{\alpha -1},a_{0},a_{1},\ldots ,a_{\alpha -2}\right) \in C_{0}$. This implies that $C_{0}$ is a cyclic code of length $\alpha$ in $\F_4[X]/\langle X^{\alpha}-1 \rangle$. 

The proof that both $C_{1}$ and $C_{2}$ are skew cyclic codes of length $\beta$ in 
$R[X]/\langle X^{\beta}-1\rangle$ follows the same line of argument. Thus, $\phi(C) = C_{0} \otimes C_{1} \otimes C_{2}$ and $\abs{\phi(C)} = 
\prod_{i=0}^{2} \abs{C_{i}}$.     
\end{proof}

\begin{lemma}\label{lemma:four}
Let $C=\left\langle(f( x),\0),(\0,g( x))\right\rangle$ be an $\F_4 R$-skew 
cyclic code with $\ell(X):=\0$. Then $C = C_{1} \otimes C_{2}$ where $C_{1}$ 
is a skew cyclic code over $\F_4$ and $C_{2}$ is a skew cyclic code over $R$.
\end{lemma}
	
\begin{proof}
Note that $\bc=(\bc_{1}, \bc_{2}) \in C $ if and only if $\bc_{1} =q_{1} f(X)$ and 
$\bc_{2}=q_{2} g(X)$ if and only if $\bc_{1} \in C_{1} = 
(f(X))$ and $\bc_{2} \in C_{2}=(g(X))$ if and only if $C=C_{1}\otimes C_{2}$ 
where $C_{1}=(f(X))$ and $\bc_{2} \in C_{2}=(g(X))$.
\end{proof}
	
\begin{lemma}\label{self-orthogonal2} 
Let $C= C_{1} \otimes C_{2}$ where $C_{1}$ is an $\F_4$-skew cyclic Euclidean self-orthogonal code 
and $C_{2}$ is an $R$-skew cyclic self-orthogonal code over $R$. Then $C$ is a self-orthogonal 
$\F_4 R$-skew cyclic code.
\end{lemma}
	
\begin{proof}
Suppose $C_{1}\subseteq C_{1}^{\perp}$ and $C_{2}\subseteq C_{2}^{\perp}$. 
Let $\bc=(\bc_{1}, \bc_{2}) \in C$ and $\bu = (\bc_{3},\bc_{4}) \in C$. 
Then $\bc_{1},\bc_{3} \in C_{1}= (f(X))$ and 
$\bc_{2},\bc_{4} \in C_{2}=(g(X))$. This implies that $\bc_{1} \cdot \bc_{3} = 0 \in \F_4$ 
and $\left\langle \bc_{2},\bc_{4}\right\rangle = 0 \in R$. 
Hence,
\[
\left\langle \bc,\bu\right\rangle = v (\bc_{1} \cdot \bc_{3})
+\left\langle \bc_{2},\bc_{4}\right\rangle =v \cdot 0 + 0 = 0.
\]
Therefore, $C\subseteq C^{\perp }$, implying that the $\F_4R$-skew cyclic code $C$ is self-orthogonal.  
\end{proof}
	
Note that the converse does not hold. In fact, $C^{\perp }\neq C_{1}^{\perp}\otimes C_{2}^{\perp }$ in general.
	
\section{$\F_4$-Codes from $\F_4R$-Skew Cyclic Codes}

\begin{table}[h!]
	\caption{Examples, in increasing minimum distances, of good choices of $\F_4$-skew cyclic codes $C_1$ and $C_2$ to use in Theorem~\ref{thm:product}.}
	\label{table:goodskew}
	\setlength{\tabcolsep}{4pt}
	\renewcommand{\arraystretch}{1.2}
	\begin{tabular}{c l l }
		\hline
		No. & Parameter & Generator Polynomial(s) for $C_1$ and $C_2$  \\
		\hline
		1 & $[6,3,4]$ & $ w + w^2 X + w^2 X^2 + X^3 $\\
		
		2 & $[30,19,7]$ & $ w^2 X + X^2 + X^4 + w^2 X^6 + X^8 + w^2 X^{10} + X^{11}$ and \\
		&          & $w + w^2 X + X^2 + X^4 + w^2 X^6 + X^8 + w^2 X^{10} + X^{11} $\\
		
		3 & $[12,3,8]$ & $w^2 + w X + w X^2 + X^3 + w^2 X^6 + w X^7 + w X^8 + X^9$ \\
		
		4 & $[22,11,8]$ & $ 1 + X + w X^2 + w X^3 + X^4 + X^5 + X^6 + X^7 + w^2 X^8 + $ \\
		& & $ w^2 X^9 + X^{10} + X^{11}$ \\
		
		5 & $[38,19,12]$ & $ 1 + X + w^2 X^2 + w^2 X^3 + w^2 X^6 + w^2 X^7 + w^2 X^8 + $ \\
		& & $ w^2 X^9 + w X^{10} + w X^{11} + w X^{12} + w X^{13} + w X^{16} + w X^{17} +$\\
		& & $ X^{18} + X^{19}$ and \\ 
		& & $1 + X + w X^2 + w X^3 + w X^6 + w X^7 + w X^8 + w X^9 + w^2 X^{10}$ \\
		& & $+ w^2 X^{11} + w^2 X^{12} + w^2 X^{13} + w^2 X^{16} + w^2 X^{17} + X^{18} + X^{19}$\\
		
		6 & $[30,6,18]$ & $w + w X + w^2 X^2 + X^3 + X^5 + w^2 X^6 + w X^7 + X^8 + $ \\
		& & $ w^2 X^9 + X^{10} + w X^{11} + w X^{12} + X^{14} + X^{15} + w^2 X^{16} + $ \\
		& & $w X^{17} + X^{18} + w^2 X^{19} + w^2 X^{20} + X^{22} + X^{23} + X^{24}$\\
		
		\hline
	\end{tabular}
\end{table}
	
The image of a code over a given ring under the Gray map is a code over a field. This latter code is usually inferior in terms of the usual measure of relative rate and relative distance when compared with linear codes directly constructed algebraically over the corresponding field. In this respect, $\F_4 R$-skew cyclic codes are not exempted. Their excellent structures, particularly as revealed in Theorem~\ref{thm:product}, allow us to determine good choices of the ingredient codes $C_0$, $C_1$, and $C_2$ that result in best-possible dimension and distance profiles. 

Since factorization in the skew polynomial ring $\F_4[X,\theta]$ requires considerably more care than in $\F_4[X]$, we started by finding good codes based on the factorization of $X^{\beta}-1 \in \F_4[X,\theta]$. A search for good skew cyclic codes was done following the suggestion of Caruso and Le Borgne in~\cite{CLB17}. Through personal communication Le Borgne sent us an implementation routine in MAGMA~\cite{BCP97}. The identified good skew cyclic codes were subsequently used as $C_1$ or $C_2$.

Based on the particular structure described in Theorem~\ref{thm:product}, the three ingredient codes $C_0,C_1$, and $C_2$ are combined by the {\tt DirectSum} routine in MAGMA to yield the code $\phi(C)$. To minimize the drop in the relative dimension and relative distance of $\phi(C)$, we chose the ingredient codes to have equal minimum distances and relatively small dimensions. Examples of good choices for the codes $C_1$ and $C_2$ are given in Table~\ref{table:goodskew} while Table~\ref{table:goodcyclic} contains the cyclic codes that we used as $C_0$. The best parameters of the resulting $\phi(C)$ are listed in Table~\ref{table:examples}.

\begin{table}[t!]
\caption{Examples of Good $\F_4$-Cyclic Codes $C_0$}
\setlength{\tabcolsep}{4pt}
\renewcommand{\arraystretch}{1.1}
\label{table:goodcyclic}
\begin{tabular}{c l l }
\hline
No. & Parameter & Generator Polynomial of $C_{0}$ \\
\hline
1 & $[4,1,4]$ & $X^3 + X^2 + X + 1$\\

2 & $[5,2,4]$ & $X^3 + w X^2 + w X + 1$ \\

3 & $[7,3,4]$ & $X^4 + X^2 + X + 1$\\

4 & $[15,11,4]$ & $X^4 + X^3 + X^2 + w^2 X + w $ \\

5 & $[17,13,4]$ & $X^4 + X^3 + w X^2 + X + 1$\\
		
6 & $[35,30,4]$ & $X^5 + w^2 X^4 + w X^2 + w X + 1$\\
	
\hline
		
7 & $[15,7,7]$ & $X^8 + X^6 + w X^5 + w X^4 + x^3 + w X^2 + w^2$ \\

8 & $[17,9,7]$ & $X^8 + w X^7 + w X^5 + w X^4 + w X^3 + w X + 1$ \\

9 & $[19,10,7]$ & $X^9 + w X^8 + w X^6 + w X^5 + w^2 X^4 + w^2 X^3 + w^2 X + 1$ \\

10 & $[35,24,7]$ & $X^{11} + w X^{10} + X^9 + w X^8 + w X^7 + w X^5 + w X^3 + X + 1$\\
		
11 & $[41,30,7]$ & $X^{11} + w X^{10} + w X^7 + w^2 X^6 + w^2 X^5 + w X^4 + w X + 1$\\
		
\hline 
		
12 & $[15,6,8]$ & $X^9 + w X^8 + X^7 + X^5 + w X^4 + w^2 X^2 + w^2 X + 1$ \\

13 &$[17,8,8]$ & $X^9 + w^2 X^8 + w X^7 + w X^6 + w X^3 + w X^2 + w^2 X + 1$ \\

14 &$[19,9,8]$ & $X^{10} + w X^9 + w^2 X^8 + w^2 X^7 + X^5 + w X^3 + w X^2 + w^2 X + 1$\\

15 & $[21,10,8]$ & $X^{11} + w X^{10} + X^8 + w^2 X^7 + w X^6 + X^5 + w X^4 + X^3 +$ \\
&& $ w X^2 + X + w$\\
		
16 & $[35,23,8]$ & $X^{12} + w^2 X^{11} + w^2 X^{10} + w^2 X^9 + w X^7 + w X^6 + w X^5 + $\\
&& $w X^4 + w X^3 + X^2 + 1$\\

17 & $[43,29,8]$ & $X^{14} + w^2 X^{13} + w^2 X^{12} + X^{11} + w^2 X^9 + w X^5 + X^3 + $ \\
&& $ w X^2 + w X + 1$\\

\hline

18 & $[15,2,12]$ & $X^{13} + w X^{12} + w X^{11} + X^{10} + X^8 + w X^7 + w X^6 + X^5 + $\\
&& $ X^3 + w X^2 + w X + 1$\\

19 &$[17,4,12]$ & $X^{13} + X^{12} + w X^{11} + X^9 + w X^8 + w^2 X^7 + w^2 X^6 + w X^5 + $ \\
&& $X^4 + w X^2 + X + 1$\\

20 &$[21,6,12]$ & $X^{15} + X^{14} + w^2 X^{13} + w^2 X^{12} + X^{11} + w X^{10} + X^9 + $ \\
&& $w^2 X^8 + X^7 + X^6 + X^5 + w X^4 + w^2 X^3 + w X^2 + wX + 1$\\

21 & $[29,14,12]$ & $X^{15} + w^2 X^{14} + w X^{13} + w X^{12} + X^{11} + w X^{10} + w X^9 + $ \\
&& $ X^8 + X^7 + w X^6 + w X^5 + X^4 + w X^3 + w X^2 + w^2 X + 1$\\

22 & $[37,18,12]$ & $X^{19} + w X^{18} + w X^{17} + w X^{15} + X^{13} + w^2 X^{12} + w X^{11} +  $\\
&& $ w X^{10} + w X^9 + w X^8 + w^2 X^7 + X^6 + w X^4 + w X^2 + w X + 1$\\

23 & $[39,19,12]$ & $X^{20} + w^2 X^{18} + w^2 X^{17} + w X^{15} + X^{14} + X^{12} + X^{11} + $\\
&& $ X^{10} + w X^9 + w^2 X^8 + w X^6 + X^5 + X^3 + w X^2 + w$\\

\hline

24 & $[43,14,18]$ & $X^{29} + w^2 X^{28} + X^{27} + w X^{24} + w^2 X^{23} 
+ X^{20} + w^2 X^{19} + $\\
&& $ w X^{17} + w^2 X^{15} + w X^{14} + w^2 X^{12} + w X^{10} + X^9 + w X^6 + $\\
&& $w^2 X^5 + X^2 + w X + 1 $ \\
\hline

\end{tabular}
\end{table}

\begin{table}[t!]
\caption{Examples of Good $\phi(C)$ where $C$ is an $\F_4 R$-Skew Cyclic Codes of Length $\alpha + \beta$. The construction is based on Theorem~\ref{thm:product}. The ingredient codes are chosen with $\delta:=d(C_0)=d(C_1)=d(C_2)$ and $\alpha$ is the length of $C_0$ while $\beta$ is the length of $C_1$ and $C_2$. The resulting $\F_4$-linear code $\phi(C)$ has length $\alpha + 2 \beta$, dimension $k_0 + k_1 + k_2$, and minimum distance $\delta$.}
\renewcommand{\arraystretch}{1.2}
\label{table:examples}
\begin{tabular}{c c | c c  | c c }
\hline
No. & $\phi(C)$ & No. & $\phi(C)$ & No. & $\phi(C)$ \\
\hline
1 & $[16,7,4]$ & 11 & $[101,68,7]$ & 21 & $[79,45,8]$ \\
2 & $[17,8,4]$ & 12 & $[39,12,8]$ & 22 & $[87,51,8]$ \\
3 & $[19,9,4]$ & 13 & $[41,14,8]$ & 23 & $[91,40,12]$ \\
4 & $[27,17,4]$ & 14 & $[43,15,8]$ & 24 & $[93,42,12]$ \\
5 & $[29,19,4]$ & 15 & $[45,16,8]$ & 25 & $[97,44,12]$ \\
6 & $[47,36,4]$ & 16 & $[59,29,8]$ & 26 & $[105,52,12]$ \\
7 & $[75,45,7]$ & 17 & $[61,30,8]$ & 27 & $[113,56,12]$ \\
8 & $[77,47,7]$ & 18 & $[63,31,8]$ & 28 & $[115,57,12]$ \\
9 & $[79,48,7]$ & 19 & $[65,32,8]$ & 29 & $[103,26,18]$ \\
10 & $[95,62,7]$ & 20 & $[67,35,8]$ &    & \\
\hline
\end{tabular}
\end{table}
	
\section{DNA Skew Cyclic Code over $\F_4R$}
	
The encoding and decoding systems to store or transfer information or data
by mimicking DNA sequences are known collectively as \textit{DNA codes}. 
The strands, \emph{i.e.}, DNA strings, are preferred to be short to
make the synthesis easy and cheap. They must, however, satisfy
numerous constraints to be useful for applications. The two most common
applications are as basic tools for biomolecular computation and as
biomolecular barcoding-tagging system to identify and manipulate individual
molecules in complex libraries.
	
Numerous approaches to DNA codes have been extensively investigated. A recent addition to several surveys that have appeared in the literature is	the work of Limbachiya \etal in~\cite{LRG}. Tools from algebraic coding theory, both from finite fields as well as rings, have been fruitfully used since the inception. A relatively early work by Marathe \etal in~\cite{Marathe2001} discussed important design criteria and bounds derived from error-correcting codes. We continue on this line of studies by constructing $\F_4R$-DNA skew cyclic codes.
	
The \textit{Watson-Crick complement} of a strand is the strand obtained by
replacing each $\mathtt{A}$ by $\mathtt{T}$ and vice versa, and each $%
\mathtt{G}$ by $\mathtt{C}$ and vice versa. One writes $\overline{\mathtt{A}}
= \mathtt{T }$, $\overline{\mathtt{T}} = \mathtt{A }$, $\overline{\mathtt{C}}
=\mathtt{G}$, and $\overline{\mathtt{G}}=\mathtt{C}$. Let $\mathbf{x}
=(x_1,x_2,\ldots,x_n)$ and $\mathbf{y}=(y_1,y_2,\ldots,y_n)$ be distinct
codewords in a DNA code $\mathcal{D}$. The \textit{reverse} of $\mathbf{x}$
is $\mathbf{x}_{\mbox{rev}}=(x_{n},x_{n-1}, \ldots, x_{1})$. The \textit{
complement} of $\mathbf{x}$ is $\mathbf{x}^{c}=(\overline{x_{1}},\overline{x_{2}},\ldots, \overline{x_n})$. Hence, $\mathbf{x}_{\mbox{rev}}^{c}= (\overline{x_{n}},\overline{x_{n-1}}, \ldots,\overline{x_{1}})$ is the \textit{reverse complement} of $\mathbf{x}$.
	
The process in which a strand and its complement bound to form a double-helix is known as \textit{hybridization}. Constraints on the codewords in a DNA code are imposed to avoid it. Let $\mathcal{D}$ be a DNA code of fixed length $n$, cardinality $M$, and minimum distance $d$. Then the constraints on the Hamming distances 
\begin{equation}  \label{eq:Hamming}
\mathrm{{wt_{H}}(\mathbf{x},\mathbf{y}) \geq d \mbox{ and } {wt_{H}}(\mathbf{x}^{c},\mathbf{y}_{\mbox{rev}}) \geq d \mbox{ for all } \mathbf{x},\mathbf{y}
\in \mathcal{D}}
\end{equation}
are imposed to prevent hybridization between any two strands as well as between a strand and the reverse of any other strand. A \textit{reverse-complement DNA code} $\mathcal{D}$ has parameters $(n,M,d)_4$ that satisfies Equation (\ref{eq:Hamming}). It is known, for instance, that any $\F_4$-cyclic code with generator polynomial $f(X)$ is reverse-complement if and only if $f(X) \in \F_4[X]$ is a self-reciprocal polynomial, \ie, $f(X)=X^{\deg(f(X))} f(X^{-1})$, not divisible by $X-1$.
	
Abualrub \etal studied $\F_4$-DNA codes of odd lengths in~\cite{AGZ06} where they use the bijection between the set of DNA alphabets $\{\mathtt{A,T,C,G\}}$ and $\F_4:=\{0,1,w,w^{2}\}$, in that respective ordering. We extend this idea by letting, for all $a\in R$, 
\begin{equation}\label{eq:bijection}
\theta(a) + \theta (\overline{a}) = v + 1.
\end{equation}
	
\begin{lemma}\label{lemma:six}
For all $a,b \in R$, we have
\begin{enumerate}[(i)]
	\item $\overline{a+b}=\overline{a}+\overline{b}+v$
	\item $\overline{(v+1)a}=(v+1)\overline{a}+v$.
	\item $\overline{va}=v\overline{a}$.
\end{enumerate}
\end{lemma}
	
The map defines the following bijection between the elements of $R$ and the $16$ codons in $\{\mathtt{A,C,G,T\}^{2}}$.
	
\begin{tabular}{cc|cc|cc|cc}
\hline
$a \in R$ & Codon & $a \in R$ & Codon & $a \in R$ & Codon & $a \in R$ & Codon
\\ \hline
$0$ & \texttt{AA} & $v+w$ & \texttt{GT} & $vw^2+1$ & \texttt{CG} & $vw$ & 
\texttt{CC} \\ 
$v$ & \texttt{TT} & $w$ & \texttt{CA} & $vw^2+w^2$ & \texttt{AG} & $vw+w^2$
& \texttt{TC} \\ 
$v+1$ & \texttt{AT} & $v+w^2$ & \texttt{CT} & $vw^2$ & \texttt{GG} & $vw+w$
& \texttt{AC} \\ 
$1$ & \texttt{TA} & $w^2$ & \texttt{GA} & $vw^2+w$ & \texttt{TG} & $vw+1$ & 
\texttt{GC} \\ \hline
\end{tabular}

\begin{definition}
An $R$-linear code $C$\ of length $\beta $ is called DNA-skew cyclic if
\begin{enumerate}
	\item The code $C$ is $R$-skew cyclic of length $\beta$. 
	\item For any codeword $\bx \in C$, $\bx \neq \bx_{\mbox{rev}}^{c}$ with $\bx_{\mbox{rev}}^{c} \in C$.
\end{enumerate}
\end{definition}
	
We adopt the following definition of reciprocal polynomials and a useful lemma from~\cite{BGM17} .
	
\begin{definition}
Let $f(X)=f_{0} +f_{1} X + \ldots + f_{k} X^{k}$ be a polynomial in $R[X,\theta]$. The reciprocal polynomial of $f(X)$ is the polynomial $f^{\star}(X)$ given by
\begin{equation}\label{eq:reciprocal}
f^{\star}(X)  
= f_{0} X^{k} + f_{1} X^{k-1} + f_{2} X^{k-2} + \ldots + f_{k-1} X+ f_{k}.
\end{equation}
If $f(X)=f(X)^{\star}$, then $f(X)$ is \emph{self-reciprocal}.
\end{definition}

\begin{lemma}\label{lemma:ast} \cite{BGM17}
Let $f(X),g(X) \in R[X,\theta]$ with $\deg (f(X)) \geq \deg (g(X))$. Then the following assertions hold.
\begin{enumerate}
	\item $(f(X)g(X))^{\star}=f(X)^{\star} \ast g(X)^{\star}$
	\item $(f(X)+g(X))^{\star}=f(X)^{\star} + g(X) \ast X^{\deg(f(X))-\deg(g(X))}$.
\end{enumerate}
\end{lemma}
	
A code is \textit{reversible complement} if $\bc_{\mbox{rev}}^c \in C$ for any $\bc \in C$. The next theorem characterizes reversible complement $R$-skew cylic code.
	
\begin{theorem}
Let two polynomials $g_{1}(X)$ and $g_{2}(X)$ divide $X^{\beta}-1$ in $\F_4[X]$. Let $C= \langle g(X)\rangle $ be $R$-skew cyclic with $g(X)=v g_{1}(X) + (v+1) g_{2}(X)$. Then $C$ is reversible complement if and only if $g(X)$ is self-reciprocal and $v (X^{\beta}-1)/(X-1) \in C$.
\end{theorem}
	
\begin{proof}
Let $g(X)=v g_{1}(X)+ (v+1) g_{2}(X)$ and $C=\langle g(X)\rangle$ be an $R$-skew cyclic code of length $\beta$. Suppose that $C$ is reversible complement. Since $\0 \in C$, we have 
$\left(\overline{0},\overline{0},\ldots,\overline{0}\right) 
=(v,v,\ldots,v) = v (X^{\beta }-1)/(X-1) \in C$. Let 
\begin{align*}
g_{1}(X) & = g_{0} + g_{1} X + \ldots + g_{t-1} X^{t-1} + X^{t} \mbox{ and }\\
g_{2}(X) & = h_{0} + h_{1} X + \ldots + h_{k-1} X^{k-1} + X^{k}
\end{align*}
where $ t \leq k$. Then
\begin{align*}
g(X) &= v g_{1}(X) + (v+1) g_{2}(X) \\
&= (v g_{0} + (v+1) h_{0}) + (v g_{1} + (v+1) h_{1}) X+ \ldots \\
& + (v g_{t-1} + (v+1) h_{t-1}) X^{t-1} + (v+ (v+1) h_{t}) X^{t} \\
& + (v+1) h_{t+1} X^{t+1}+ \ldots + (v+1) h_{k-1} X^{k-1} + (v+1) X^{k}.
\end{align*}
Since $C$ is reversible complement, it contains   
\begin{align*}
g_{\mbox{rev}}^{c} (X) 
& = v (1 + X + \ldots + X^{\beta-k-2}) + \overline{ (v+1)} 
X^{\beta- k -1 } + \overline{(v+1) h_{k-1}} X^{\beta-k} \\
& + \ldots + \overline{(v+1) h_{t+1}} X^{\beta -t-2} + 
\overline{(v+(v+1) h_{t})} X^{\beta-t-1} \\
& + \overline{(v g_{t-1} + (v+1) h_{t-1})} X^{\beta-t} + \ldots + 
\overline{(v g_{1}  + (v+1) h_{1})} X^{\beta-2} \\
& + \overline{(v g_{0} + (v+1) h_{0})} X^{\beta-1}.
\end{align*}
Using Lemma~\ref{lemma:six} we can write 
\begin{align*}
g_{\mbox{rev}}^{c} (X) & =
v (1+ X + \ldots + X^{\beta-k-2})+ \overline{(v+1)} X^{\beta -k-1} + \overline{(v+1) h_{k-1}} X^{\beta -k} \\ 
& + \ldots + \overline{(v+1) h_{t+1}} X^{\beta-t-2} + \overline{v} 
X^{\beta-t-1} + \overline{(v+1) h_{t}} X^{\beta-t-1} + v X^{\beta-t-1} \\
& + \overline{v g_{t-1}} X^{\beta-t} + \overline{(v+1) h_{t-1}} X^{\beta-t} + v X^{\beta-t} + \ldots + \overline{v g_{1}}X^{\beta-2} \\
&+ \overline{(v+1) h_{1}} X^{\beta -2} + v x^{\beta-2} + \overline{v g_{0}} X^{\beta-1} + \overline{(v+1) h_{0}} X^{\beta-1} + v X^{\beta-1}.
\end{align*}
Because $C$ is $R$-linear, $g_{\mbox{rev}}^{c} (X) + v (X^{\beta}-1) / (X-1) \in C$. This implies 
\begin{align*}
& g_{\mbox{rev}}^{c} (X) + v (X^{\beta }-1)/(X-1) 
= \\
&(\overline{(v+1)} + v) X^{\beta-k-1} + (\overline{(v+1) h_{k-1}} + v) X^{\beta-k} + \ldots + \\
& (\overline{(v+1) h_{t+1}} + v) X^{\beta-t-2} + 
(\overline{v} + v) X^{\beta-t-1+} (\overline{(v+1) h_{t}} + v) 
X^{\beta-t-1} + \\
& (\overline{v g_{t-1}} + v) X^{\beta-t}  + (\overline{(v+1) h_{t-1}} + v) 
X^{\beta-t} + \ldots + (\overline{v g_{1}} + v) X^{\beta-2} + \\
& (\overline{(v+1) h_{1}} + v) X^{\beta -2} + (\overline{v g_{0}} + v) X^{\beta-1} + (\overline{(v+1) h_{0}} + v) X^{\beta-1}.
\end{align*}
By Equation~(\ref{eq:bijection}) we can write
\begin{align*}
& (v+1) X^{\beta-k-1} + (v+1) {h_{k-1}} X^{\beta-k} + \ldots + (v+1) {h_{t+1}} X^{\beta-t-2} +\\
& v X^{\beta-t-1}+ (v+1) {h_{t}} X^{\beta-t-1} + 
X^{\beta-t} + (v+1) {h_{t-1}}X^{\beta-t} + \ldots + \\
& v {g_{1}} X^{\beta-2} + (v+1) {h_{1}} X^{\beta-2}+ v {g_{0}} X^{\beta-1} + (v+1) {h_{0}} 
X^{\beta-1}
\end{align*}
as
\begin{align*}
	& (v+1)  X^{\beta-k-1} + (v+1) {h_{k-1}} X^{\beta-k} + \ldots + (v+1) {h_{t+1}} X^{\beta-t-2} + \\
	&(v + (v+1) {h_{t}}) X^{\beta-t-1} + (v {g_{t-1}} + (v+1) {h_{t-1}}) X^{\beta-t}+ \ldots + \\
	&(v {g_{1}} + (v+1) {h_{1}}) X^{\beta-2} + 
	(v {g_{0}} + (v+1) {h_{0}}) X^{\beta-1}.
\end{align*}
Multiplying on the right by $X^{k + 1 -\beta}$, we obtain
\begin{multline*}
(v+1) + (v+1) {h_{k-1}}\theta(1) X + \ldots + (v+1) {h_{t+1}}\theta(1) X^{k-t-1}  \\
+ (v+(v+1) {h_{t}})\theta(1) X^{k-t} + (v {g_{t-1}} + (v+1) {h_{t-1}})\theta(1) X^{k-t+1} + \ldots   \\
+ (v {g_{1}} + (v+1) {h_{1}})\theta(1) X^{k-1} + (v {g_{0}} + (v+1) {h_{0}})\theta(1) X^{k}.
\end{multline*}
Hence, $g^{\star}(X) \in C$. Since $C = \langle g(X) \rangle$, there exists $q(X) \in R[X,\theta]$ such that $g^{\star}(X)= q(X) \ast g(X)$, which implies $\deg (g^{\star}(X)) = \deg (g(X))$ and $q(X)=1$. Thus, $g^{\star}(X) =g(X)$, as required.
		
Conversely, let $C=\langle g(X)\rangle $ be an $R$-skew cyclic code of length $\beta$ generated by $g(X) = v g_{1}(X) + (v+1) g_{2}(X)$ where $g_{1}(X)$ and $g_{2}(X)$ are two divisors of $X^{\beta}-1$ in $\F_4[X]$. Let $c(X) = c_{0} + c_{1} X + \ldots + c_{k} X^{k} \in C$, then there exist $q(X) \in R[X,\theta]$ such that $c(X) = q(X) \ast g(X)$. By Lemma~\ref{lemma:ast}, $c^{\star}(X) =q^{\star}(X) \ast g^{\star}(X)$. Since $C$ is self reciprocal, $c^{\star}(X) = q^{\star}(X) \ast g(X) \in C$ for any $c(X)\in C$. Since $C$ is skew cylic, $c(X) \ast X^{\beta-k-1} = c_{0} X^{\beta-k-1} + c_{1} X^{\beta-k} + \ldots + c_{k} X^{\beta-1} \in C$. Hence, $v (X^{\beta}-1) / (X-1) = 
v (1 + \ldots + X^{\beta-1}) \in C$. Since $C$ is $R$-linear, 
\begin{multline*}
c(X) \ast X^{\beta-k-1} + v (X^{\beta}-1) / (X-1) =\\
v + \ldots + v X^{\beta-k-2} + (c_{0} + v) X^{\beta-k-1} + \ldots + (c_{k} + v) X^{\beta-1}.
\end{multline*}
By Equation~(\ref{eq:bijection}), 
\[
v + \ldots + v X^{\beta-k-2} + \overline{c_{0}} 
X^{\beta-k-1} + \ldots + \overline{c_{K}} X^{\beta-1}
=(c^{\ast}(X))_{\mbox{rev}}^c \in C.
\]
This concludes the proof.  
\end{proof}
	
The theorem that we have just proved leads us from $R$-skew cyclic code to the definition and subsequent characterization of $\F_4 R$-skew cyclic code in the context of DNA coding.

\begin{definition}
An $\F_4R$-linear code $C$ is DNA-skew cyclic if the followings hold.
\begin{enumerate}
	\item $C$ is an $\F_4R$-skew cyclic code, \ie, $C$ is an $R$-left submodule of 
	\[
	F[X]/(X^{\alpha }-1) \times R[X,\theta ]/\left(X^{\beta }-1\right).
	\]
	
	\item Any codeword $\bc=(\bc_{1},\bc_{2})\in C$ and its reverse complement 
	\[
	\bc_{\mbox{rev}}^{c}=((\bc_{1})_{\mbox{rev}}^{c},(\bc_{2})_{\mbox{rev}}^{c}) \in C
	\]
	must be distinct.
\end{enumerate}
\end{definition}
	
The characterization of reverse complement codes over $\F_4R$ can now be established. 
\begin{theorem}
Let $C=\langle (f(X),\0),(\0,g(X) \rangle$ be an $\F_4R$-skew cyclic code. Note that $\ell(X):=\0$ and $C=C_{1} \otimes C_{2}$ with $C_{1}$ an $\F_4$-cyclic code and $C_{2}$ an $R$-skew cyclic code. Then $C$ is reversible complement if and only if $C_{1}$ and $C_{2}$ are reversible complement over $\F_4$ and $R$, respectively.
\end{theorem}
	
\begin{proof}
Let $C$ be an $\F_4R$-skew cyclic code generated by $\left(f(X),\0\right)$ and $(\0,g(X))$. Lemma~\ref{lemma:four} shows how to find $C=C_{1} \otimes C_{2}$. Let $\bc=(\bc_{1},\bc_{2})\in C = C_{1} \otimes C_{2}$ with $\bc_{1}\in C_{1}$ and $\bc_{2}\in C_{2}$. Suppose that $C_{1}$ and $C_{2}$ are reversible complement over $\F_4$ and $R$, respectively. Then we have $(\bc_{1})_{\mbox{rev}}^{c} \in C_{1}$ and $(\bc_{2})_{\mbox{rev}}^{c} \in C_{2}$. Thus, $((\bc_{1})_{\mbox{rev}}^{c}, (\bc_{2})_{\mbox{rev}}^{c})= \bc_{\mbox{rev}}^{c} \in C_{1} \otimes C_{2}=C$.
		
Conversely, let $\bc_{1} \in C_{1}$ and $\bc_{2} \in C_{2}$. Then  $\bc=(\bc_{1},\bc_{2}) \in C$. If $C$ is reversible complement, then   $\bc_{\mbox{rev}}^{c}=((\bc_{1})_{\mbox{rev}}^{c},(\bc_{2})_{\mbox{rev}}^{c}) \in C = C_{1} \otimes C_{2}$. This implies $\bc_{1,\mbox{rev}}^{c}\in C_{1}$ and $\bc_{2,\mbox{rev}}^{c} \in C_{2}$, as required. 
\end{proof}

\section{Conclusion}\label{sec:conclu}

We have presented our studies on skew cyclic codes over the ring $\F_4R$. Their  algebraic structures as left submodules of a skew-polynomial ring are investigated, resulting in the identification of their generators. The fact that, under some simple conditions on their length, they are equivalent to cyclic or $2$-quasi-cyclic codes over the same ring is established. Towards the end we show how the setup leads naturally to DNA codes and prove a condition on the associated generator polynomial of an $\F_4R$-skew cyclic code that guarantee the code to be reversible complement. We are now looking into whether the class of codes that we propose here contains those with better relative distance or size than known DNA codes.

\section*{Acknowledgments}
We would like to thank Jeremy Le Borgne for the MAGMA source code for factorization in $\F_q[X,\theta]$.

\medskip
Received xxxx 20xx; revised xxxx 20xx.
\medskip


\begin{thebibliography}{10}

\bibitem{Abualrub2010} 
	\newblock T. Abualrub, A. Ghrayeb, N. Aydin and I. Siap, 
	\newblock On the construction of skew quasi-cyclic codes, 
	\newblock \emph{IEEE Trans. Inform. Theory}, \textbf{56} (2010), 2081--2090.

\bibitem{AGZ06} 
	\newblock T. Abualrub, A. Ghrayeb and X. N. Zeng,
	\newblock Construction of cyclic codes over $GF(4)$ for DNA computing, 
	\newblock \emph{J. Franklin Inst.}, \textbf{343} (2006), 448--457.

\bibitem{Abualrub2014} 
	\newblock T. Abualrub, I. Siap and N. Aydin, 
	\newblock $\Z_{2}\Z_{4}$-additive cyclic codes, 
	\newblock \emph{IEEE Trans. Inform. Theory}, \textbf{60} (2014), 1508--1514.
	
\bibitem{Bayram} 
	\newblock A. Bayram, E. Oztas and I. Siap, 
	\newblock Codes over $\F_{4}+v \F_{4}$ and some DNA applications, 
	\newblock \emph{Des. Codes Cryptogr.}, \textbf{80} (2016), 379--393.
	
\bibitem{BGM17}
	\newblock N. Bennenni, K. Guenda and S. Mesnager, 
	\newblock DNA cyclic codes over rings, 
	\newblock \emph{Adv. in Math. of Comm.}, \textbf{11} (2017), 83--98.
	
\bibitem{Borges2010} 
	\newblock J. Borges, C. Fern\'{a}ndez-C\'{o}rdoba, J. Pujol, J. Rif\`{a} and M. Villanueva, 
	\newblock $\Z_{2}\Z_{4}$-linear codes: Generator matrices and duality, 
	\newblock \emph{Des. Codes Cryptogr.}, \textbf{54} (2010), 167--179.
	
\bibitem{BCP97}
	\newblock W.~Bosma, J.~Cannon and C.~Playoust, 
	\newblock The Magma algebra system. I: The user language,
	\newblock \emph{J. Symbolic Comput.}, \textbf{24} (1997) pp. 235--265.
1997.

\bibitem{Boucher2007} 
	\newblock D. Boucher, W. Geiselmann and F. Ulmer,
	\newblock Skew-cyclic codes, 
	\newblock \emph{Appl. Algebra Engrg. Comm. Comput.}, \textbf{18} (2007), 379--389.
	
\bibitem{Boucher2008} 
	\newblock D. Boucher, P. Sole and F. Ulmer, 
	\newblock Skew constacyclic codes over Galois rings,
	\newblock \emph{Adv. Math. Comm.}, \textbf{2} (2008), 273--292.
	
\bibitem{CLB17}
	\newblock X. Caruso and J. Le Borgne,
	\newblock A new faster algorithm for factoring skew polynomials over finite fields,
	\newblock \emph{J. Symbolic Comput.}, \textbf{79} (2017), 411--443 
	
\bibitem{Gursoy2014} 
	\newblock F. Gursoy, I. Siap and B.~Yildiz, 
	\newblock Construction of skew cyclic codes over $\F_{q} + v \F_{q}$, 
	\newblock \emph{Adv. Math. Comm.}, \textbf{8} (2014), 313--322.
	
\bibitem{LRG}
	\newblock D. Limbachiya, B. Rao and M. K. Gupta, 
	\newblock The art of DNA strings: Sixteen years of DNA coding theory. 
	\newblock CoRR \textbf{abs/1607.00266} (2016).
	\url{http://arxiv.org/abs/1607.00266}
	
\bibitem{Marathe2001}
	\newblock A. Marathe, A. E. Condon and R. M. Corn, 
	\newblock On combinatorial DNA word design, 
	\newblock \emph{J. Comput. Biology} \textbf{8} (2001), 201--219.
	
\bibitem{ZWS10} 
	\newblock S. Zhu, Y. Wang, and M. Shi, 
	\newblock Some result on cyclic codes over $\F_{2} + v \F_{2}$, 
	\newblock \emph{IEEE Trans. Inform. Theory}, \textbf{56} (2010), 1680--1684.
	
\end{thebibliography}
\end{document}